\documentclass{article}[11pt]

\usepackage{amsmath,amssymb,amsthm}
\usepackage{url}
\usepackage{fullpage}

\newcommand{\eps}{\varepsilon}
\newcommand{\R}{\mathbb{R}}

\newcommand{\EquationName}[1]{\label{eq:#1}}

\newcommand{\CorollaryName}[1]{\label{cor:#1}}
\newcommand{\SectionName}[1]{\label{sec:#1}}
\newcommand{\TheoremName}[1]{\label{thm:#1}}

\newcommand{\Equation}[1]{Eq.\:\eqref{eq:#1}}

\newcommand{\Corollary}[1]{Corollary~\ref{cor:#1}}
\newcommand{\Section}[1]{Section~\ref{sec:#1}}
\newcommand{\Theorem}[1]{Theorem~\ref{thm:#1}}

\newcommand{\Eqsub}[1]{\eqref{eq:#1}}

\newtheorem{theorem}{Theorem}
\newtheorem{lemma}[theorem]{Lemma}
\newtheorem{fact}[theorem]{Fact}
\newtheorem{corollary}[theorem]{Corollary}
\def\reals{\mathbb{R}}
\def\tail{\mathrm{tail}}
\newcommand{\eqdef}{\mathbin{\stackrel{\rm def}{=}}}

 \author{Tom Morgan\thanks{Harvard University. \texttt{tdmorgan@seas.harvard.edu}. Supported by NSF grants CCF-1320231 and CNS-1228598.}
   \and Jelani Nelson\thanks{Harvard University. \texttt{minilek@seas.harvard.edu}. Supported by NSF grant IIS-1447471 and
   CAREER award CCF-1350670, ONR Young Investigator award N00014-15-1-2388, and a Google Faculty Research Award.}
}

\title{A note on reductions between compressed sensing guarantees}

\begin{document}

\maketitle

\begin{abstract}
In {\em compressed sensing}, one wishes to acquire an approximately sparse high-dimensional signal $x\in\R^n$ via $m\ll n$ noisy linear measurements, then later approximately recover $x$ given only those measurement outcomes. Various guarantees have been studied in terms of the notion of approximation in recovery, and some isolated folklore results are known stating that some forms of recovery are stronger than others, via black-box reductions. In this note we provide a general theorem concerning the hierarchy of strengths of various recovery guarantees. As a corollary of this theorem, by reducing from well-known results in the compressed sensing literature, we obtain an efficient $\ell_p/\ell_p$ scheme for any $0<p<1$ with the fewest number of measurements currently known amongst efficient schemes, improving recent bounds of \cite{SomaY16}.
\end{abstract}

\section{Introduction}\SectionName{intro}

The field of {\em compressed sensing} \cite{CandesT05, Donoho06} is concerned with recovering approximately sparse signals from few (possibly noisy) linear measurements. That is, given access to $y = \Phi x + e$, where $\Phi\in\R^{m\times n}$ is some matrix and $x, e\in\R^n$ has $e$ of small norm and $x$ being approximately sparse. That is, $x$ can be decomposed as $x = f + e'$ where $f$ is sparse, i.e.\ $k \eqdef \|f\|_0$ is small, and $e'$ has small norm. Here $\|\cdot\|_0$ denotes support size. Ideally we would like $m\ll n$ (i.e.\ few measurements), and that there is an efficient algorithm $\mathcal{R}$ which, knowing $\Phi$ and given only access to $y$ but not $x$, recovers some $\hat{x} = \mathcal{R}(y)$ such that $x-\hat{x}$ has small norm (in terms of the norms of $e, e'$). That is, for some norms $\|\cdot\|_X, \|\cdot\|_Y$ and some constants $C>0$ and $c\ge 0$ depending only $X, Y$, we would like
\begin{equation}
\|x - \mathcal{R}(y)\|_X \le C k^{-c} \cdot \sigma_k(x)_Y , \EquationName{uniform}
\end{equation}
where $\sigma_k(x)_Y = \inf_{\|z\|_0 \le k} \|x - z\|_Y$ is the $\|\cdot\|_Y$-norm error of the best $k$-sparse approximation to $x$. Popular guarantees investigated previously include so-called $\ell_p/\ell_q$ guarantees, of the form
\begin{equation}
\|x - \mathcal{R}(y)\|_{\ell_p^n} \le C k^{-c} \cdot \sigma_k(x)_{\ell_q^n}  \EquationName{guarantee}
\end{equation}
where the constant $c$ equals $1/q - 1/p$ \cite{CohenDD08}. For the case $p = q$, one can often even take $C = 1+\eps$ and develop schemes with small $m$ and efficient $\mathcal{R}$, in which a dependence on $\eps$ enters $m$ (see e.g.\ \cite{IndykR08} for $p=q=1$). In all the cases of $1\le q\le p<\infty$, it is known that any scheme achieving \Eqsub{guarantee} must have $m \gtrsim k\log(n/k)$ (and sometimes much larger, depending on $p, q$), regardless of how inefficient $\mathcal{R}$ is allowed to be. This is due to a connection with known bounds on Gelfand widths \cite{CohenDD08}. 

The case $p = q < 1$ was first investigated in \cite{ChartrandS08,ShenL10} (although $\ell_p$ is not a norm in this case), where it was shown that there exists {\em some} recovery scheme $\mathcal{R}$ (though not efficient) that allows for $m\simeq C_1(p)k + C_2(p)k\log(n/k)$ measurements for some $C_2(p)\rightarrow 0$ as $p\rightarrow 0$. In other words, for sufficiently small $p \ll 1$, there exists an upper bound violating the lower bound for $\ell_q/\ell_q$ norm guarantees for $q\ge 1$. For the case $p = 1/2^j$ where $j$ is any positive integer, and for any fixed constant $\epsilon>0$, recent work of \cite{SomaY16} provides a scheme with an {\em efficient} (polynomial-time in $n$) $\mathcal{R}$ for achieving a slight weakening of the $\ell_p/\ell_p$ recovery guarantee, but with a larger number of measurements $m \simeq k^{2/p}\log n$. Specifically, their error guarantee is 
$$
\|x - \mathcal{R}(y)\|_{\ell_p^n} \le C \sigma_k(x)_{\ell_q^n} + \epsilon\cdot\|x\|_\infty
$$

It is folklore that some norm guarantees are stronger than others. By stronger, we mean a reduction in the sense that if the scheme $(\Phi, \mathcal{R}_1)$ achieves the $\|\cdot\|_{X_1} / \|\cdot\|_{Y_1}$ recovery guarantee, then there is an efficient algorithm $\mathcal{A}$ such that if $\mathcal{R}_2(y)$ is simply set to be $\mathcal{A}(\mathcal{R}_1(y), y)$, then $(\Phi, \mathcal{R}_2)$ achieves the $\|\cdot\|_{X_2} / \|\cdot\|_{Y_2}$ guarantee (sometimes in these reductions, the values of $k$ and $C$ may change by constant factors). In this sense, achieving the $\|\cdot\|_{X_1} / \|\cdot\|_{Y_1}$ guarantee is {\em stronger} than achieving the $\|\cdot\|_{X_2} / \|\cdot\|_{Y_2}$ guarantee (since devising an efficient scheme for the former implies an efficient scheme for the latter). It is folklore that, for example, an $\ell_2/\ell_2$ scheme is stronger than $\ell_2/\ell_1$, which in turn is stronger than $\ell_1/\ell_1$ (note though it is impossible to achieve the $\ell_2/\ell_2$ guarantee without weakening to a certain probabilistic guarantee \cite{CohenDD08}, i.e.\ {\em nonuniformity} --- we discuss nonuniformity in \Section{nonuniform}). The main observation of this note is a common generalization of both of these folklore reductions. In particular, we show the following main theorem.

\begin{theorem}\TheoremName{main}
For $p \geq r \geq s > 0$ and $q \geq s$, if there exists a recovery scheme $(\Phi, \mathcal{R}')$ s.t.\ for some constant $C'$ independent of $k, n$
$$\|x-\mathcal{R}'(y)\|_p \leq C'\cdot k^{\frac{1}{p}-\frac{1}{q}} \sigma_{2k}(x)_q ,$$
then there exists a recovery scheme $(\Phi, \mathcal{R})$ such that for some constant $C$
$$\|x-\mathcal{R}(y)\|_r \leq C \cdot k^{\frac{1}{r}-\frac{1}{s}} \sigma_k(x)_s .$$
Furthermore, if $\mathcal{R}'$ runs in time $T$ and outputs a vector $\hat{x}$ of support size $S$, then $\mathcal{R}$ runs in time $O(T + S)$.
\end{theorem}

The reduction of \Theorem{main} is very simple: the recovery algorithm $\mathcal{R}$, given $\Phi x$, first computes $\mathcal{R}'(\Phi x)$. We then let $\mathcal{R}(\Phi x)$ be defined by projecting $\mathcal{R}'(\Phi x)$ to its largest $2k$ coordinates (in magnitude). This scheme is analyzed in \Section{main} (see \Theorem{main-thm}).

It is known how to achieve the $\ell_2/\ell_1$ guarantee with $m \simeq k\log(n/k)$ measurements and a $\mathop{poly}(n)$-time recovery algorithm $\mathcal{R}$ \cite{Candes08} (more accurately, there exists a deterministic algorithm $\mathcal{R}$ such that if $\Phi$ is drawn at random from a particular distribution with that number $m$ of measurements, then with probability $1 - \mathop{poly}(1/\binom{n}{k})$, $(\Phi, \mathcal{R})$ is a scheme achieving the $\ell_2/\ell_1$ recovery guarantee). Thus by setting $r = s = p\le 1$ in \Theorem{main}, we obtain the following corollary.

\begin{corollary}\CorollaryName{small-p}
For any $0<p\le 1$, there exists a scheme achieving the $\ell_p/\ell_p$ recovery guarantee with $m \lesssim k\log(n/k)$ measurements and recovery time $\mathop{poly}(n)$.
\end{corollary}

We point out that although \Corollary{small-p} does not achieve the $C_1(p)k + C_2(p)k\log(n/k)$ measurements for $C_2(p)\rightarrow 0$ that was shown achievable (with an inefficient recovery algorithm) in \cite{ChartrandS08,ShenL10}, it achieves a number of measurements that is much less than the $m \simeq k^{2/p}\log n$ of \cite{SomaY16}, even for $p=1$.

\subsection{Nonuniform guarantees}\SectionName{nonuniform}

Up until this point we have only discussed {\em uniform} recovery guarantees (this terminology appears in e.g.\ \cite{FoucartR13}). A uniform scheme is one where \Equation{uniform} holds for all $x\in\R^n$ simultaneously, for single pair $(\Phi, \mathcal{R})$. Indeed many such schemes pick $\Phi$ randomly, but then the desired guarantee for uniformity is
$$
\Pr_{\Phi}(\forall x\in\R^n,\ \Eqsub{uniform}\text{ holds}) \ge 1-\delta .
$$
A {\em nonuniform} recovery guarantee is one where, in a randomized scheme (in which $\Phi$, and possibly also $\mathcal{R}$, are chosen at random from some distribution)
$$
\forall x\in\R^n \Pr_{\Phi, \mathcal{R}}(\Eqsub{uniform}\text{ holds}) \ge 1-\delta .
$$
Uniform and nonuniform schemes are also sometimes called ``for all'' and ``for each'' schemes in the literature. It is known for example that the $\ell_2/\ell_2$ guarantee cannot be achieved uniformly by {\em any} (even possibly inefficient) recovery algorithm unless $m \gtrsim n$ \cite{CohenDD08}, however it is achievable by a nonuniform scheme with failure probability $1/\mathop{poly}(n)$, $m\simeq k\log n$, recovery time $T \lesssim n\log n$, and output sparsity $S = \|\hat{x}\|_0 \lesssim k$ by combining the CountSketch \cite{CharikarCF04} with a reduction from $\ell_2/\ell_2$ recovery to the $\ell_2$ heavy hitters problem \cite{CormodeM06} (see also \cite[Section II.B]{GilbertI10}). One could also replace the CountSketch with the more efficient ExpanderSketch to keep all parameters the same while reducing the recovery time to $T \lesssim k\log^c n$. We note that to achieve $C = 1+\eps$, the work of \cite{GilbertLPS10} achieves a better bound on $m$ than the ExpanderSketch in terms of $\eps$ by a factor of $1/\eps$, albeit with a worse value of $S$ and only constant failure probability. However, since \Theorem{main} only preserves $C$ up to a constant factor, it cannot be used to convert a scheme with $C = 1+O(\eps)$ for one recovery guarantee into a scheme with $C = 1+\eps$ for another guarantee.

We now mention that \Theorem{main} holds regardless of whether $(\Phi, \mathcal{R}')$ is uniform or nonuniform scheme, and the reduction is uniformity-preserving. Thus by combining the observations of the last paragraph with \Theorem{main}, we obtain the following corollary.

\begin{corollary}\CorollaryName{nonuniform-small-p}
For any $0<p\le 1$, there exists a nonuniform scheme achieving the $\ell_p/\ell_p$ recovery guarantee with $m \lesssim k\log n$ measurements, recovery time $k\cdot \mathop{poly}(\log n)$, and failure probability $1/\mathop{poly}(n)$.
\end{corollary}
\section{Main result}\SectionName{main}

Given $v \in \reals^n$ and $S \subseteq [n]$ we will use $v_S$ to denote the vector where $(v_S)_i = v_i$ if $i \in S$ and 0 otherwise.

\begin{fact} \label{fact:lp_ba}
For any $v \in \reals^n$ and $0 < a \leq b$, $\|v\|_b \leq \|v\|_a$.
\end{fact}

\begin{proof}
Observe that for all $i \in [n]$, $\frac{|v_i|}{\|v_i\|_b} \leq 1$.  This, together with the fact that $a \leq b$, gives us
\begin{equation} \label{eq:normalized}
\left(\frac{|v_i|}{\|v\|_b}\right)^b \leq \left(\frac{|v_i|}{\|v\|_b}\right)^a.
\end{equation}

We now have
\begin{align*}
\|v\|_a &= \left(\sum_{i=1}^n |v_i|^a\right)^\frac{1}{a}\\
&\geq \left(\sum_{i=1}^n \left(\frac{|v_i|}{\|v\|_b}\right)^b \|v\|_b^a\right)^\frac{1}{a}\text{ (\Equation{normalized})}\\
&= \|v\|_b^{1-\frac{b}{a}} \left(\sum_{i=1}^n |v_i|^b \right)^\frac{1}{a} \\
&= \|v\|_b.
\end{align*}
\end{proof}

\begin{fact} \label{fact:lp_ab}
For any $v \in \reals^n$ and $0 < a \leq b$, $\|v\|_a \leq n^{\frac{1}{a}-\frac{1}{b}} \|v\|_b$. In particular, if $v$ is $\kappa$-sparse then $\|v\|_a \leq \kappa^{\frac{1}{a}-\frac{1}{b}} \|v\|_b$.
\end{fact}

\begin{proof}
H\"{o}lder's inequality states that for $p \geq 1$,
$$\|fg|\|_1 \leq \|f\|_p \|g\|_{\frac{p}{p-1}}.$$

We will choose $f,g \in \reals^n$ such that for all $i \in [n]$, $|f_i| = |v_i|^a$ and $|g_i| = 1$.  Letting $p = b/a$, we then have
$$\|v\|_a^a = \|f g\|_1 \leq \|f\|_{\frac{b}{a}} \|g\|_{\frac{b}{b-a}} = \|v\|_b^{a} \cdot n^{1-\frac{a}{b}},$$
and thus
$$\|v\|_a \leq \|v\|_b \cdot n^{\frac{1}{a}-\frac{1}{b}}$$
as desired.
\end{proof}

The following lemma is via a common technique in compressed sensing that has come to be known as {\em shelling}, in which one sorts coordinates of a vector by magnitude of entries, blocks consecutive groups of coordinates together, then compares some norm of one group with some norm of the previous group.

\begin{lemma} \label{lem:shelling}
For any $v \in \reals^n$ and $0 < a \leq b$,
$$\|v_{\tail(2\kappa)}\|_b \leq \kappa^{\frac{1}{b} - \frac{1}{a}} \|v_{\tail(\kappa)}\|_a.$$
\end{lemma}

\begin{proof}
Let $i_1,\ldots,i_n$ be a permutation of $[n]$ such that $|v_{i_1}| \le |v_{i_2}| \le \ldots \le |v_{i_n}|$. For $j\ge 0$ an integer define $B_j = \{i_{jk+1}, i_{jk+2}, \ldots, i_{(j+1)k}\}$. Then $v_{tail(2\kappa)}$ is simply $v_{B\backslash (B_0\cup B_1)}$, so that
\begin{equation}
\|v_{tail(2\kappa)}\|_b = \left(\sum_{j\ge 2} \|v_{B_j}\|_b^b \right)^{\frac{1}{b}}. \label{eq:about-to-shell}
\end{equation}
Next observe that for any $i\in B_j$, $|v_i| \le \|v_{B_{j-1}}\|_a/\kappa^{1/a}$. This follows since $\|v_{B_{j-1}}\|_a^a/\kappa$ is the average $a$th power of all $|v_{i'}|$ in $B_{j-1}$, and all terms in this average are at least as big as $|v_i|^a$. Combining with \Equation{about-to-shell}
\begin{align*}
\|v_{\tail(2\kappa)}\|_b &\leq \left(\sum_{j\ge 2} \frac{\kappa}{\kappa^{b/a}}\|v_{B_{j-1}}\|_a^b\right)^{\frac{1}{b}}\\
{}&= \kappa^{\frac{1}{b} - \frac{1}{a}} \cdot \left(\sum_{j\ge 1} \|v_{B_j}\|_a^b\right)^{\frac{1}{b}}\\
{}&\le \kappa^{\frac{1}{b} - \frac{1}{a}} \cdot \left(\sum_{j\ge 1} \|v_{B_j}\|_a^a\right)^{\frac{1}{a}} \text{ (Fact~\ref{fact:lp_ba})}\\
{}&= \kappa^{\frac{1}{b} - \frac{1}{a}} \cdot \|v_{\tail(\kappa)}\|_a.
\end{align*}
\end{proof}

Now we have all the tools to prove our main theorem.

\begin{theorem}\TheoremName{main-thm}
For $p \geq r \geq s > 0$ and $q \geq s$, if there exists a recovery scheme $(\Phi, \mathcal{R}')$ such that for some constant $C'_{p,q}$
$$\|x-\mathcal{R}'(\Phi x)\|_p \leq C'_{p,q} \cdot k^{\frac{1}{p}-\frac{1}{q}} \|x_{\tail(2k)}\|_q,$$
then there exists a recovery scheme $(\Phi, \mathcal{R})$ such that for some constant $C_{p,q,r}$
$$\|x-\mathcal{R}(\Phi x)\|_r \leq C_{p,q,r} \cdot k^{\frac{1}{r}-\frac{1}{s}} \|x_{\tail(k)}\|_s.$$
Furthermore, if $(\Phi, \mathcal{R}')$ achieves a uniform guarantee, then so does $(\Phi,\mathcal{R})$. Also, if $\mathcal{R}'$ runs in time $T$ and outputs a vector of support size $S$, then $\mathcal{R}$ runs in time $O(T + S)$.
\end{theorem}
\begin{proof}
The recovery algorithm $\mathcal{R}$, given $\Phi x$, first computes $w = \mathcal{R}'(\Phi x)$. We then let $z = \mathcal{R}(\Phi x)$ be defined by projecting $w$ to its largest $2k$ coordinates (in magnitude).  It is clear that this reduction is uniformity-preserving, and it runs in time $O(T+S)$ since the $(2k)^{\mbox{th}}$ largest element of $w$ (in magnitude) can be found in $O(S)$ time using the linear-time selection algorithm of \cite{BlumFPRT72}.

We now analyze this scheme. Let $A\subseteq[n]$ be the largest $2k$ coordinates of $x$ in magnitude, and $B\subseteq[n]$ be the largest $2k$ coordinates of $w$ in magnitude (i.e. $z=w_B$).

First we will prove the case when $r \leq 1$.  We will use the fact that $\ell_r$ for $r \leq 1$ is a quasinorm, with $d(f,g) := \|f-g\|_r^r$ satisfying the triangle inequality.

\begin{align*}
\|x - z\|_r^r &\leq \|x_B - w_B\|_r^r + \|x - x_B\|_r^r\text{ (triangle inequality)}\\
{}&= \|x_B - w_B\|_r^r + \|x\|_r^r - \|x_B\|_r^r\\
{}&\leq \|x_B - w_B\|_r^r + \|x\|_r^r - (\|w_B\|_r^r - \|x_B - w_B\|_r^r)\text{ (triangle inequality)}\\
{}&= 2 \|x_B - w_B\|_r^r + \|x\|_r^r - \|w_B\|_r^r\\
{}&\leq  2 \|x_B - w_B\|_r^r + \|x\|_r^r - \|w_A\|_r^r\text{ (definition of }B\text{)}\\
{}&\leq  2 \|x_B - w_B\|_r^r + \|x\|_r^r - (\|x_A\|_r^r - \|w_A - x_A\|_r^r) \text{ (triangle inequality)}\\
{}&=  2 \|x_B - w_B\|_r^r + \|w_A - x_A\|_r^r + \|x_{\tail(2k)}\|_r^r\\
{}&\leq 2(2k)^{1-\frac{r}{p}} \|x_B-w_B\|_p^r + (2k)^{1-\frac{r}{p}} \|w_A - x_A\|_p^r + \|x_{\tail(2k)}\|_r^r \text{ (Fact~\ref{fact:lp_ab})}\\
{}&\leq 3(2k)^{1-\frac{r}{p}}\|x - w\|_p^r + \|x_{\tail(2k)}\|_r^r\\
{}&\leq 3 \cdot 2^{1-\frac{r}{p}} \left(C'_{p,q}\right)^r k^{1-\frac{r}{q}} \|x_{\tail(2k)}\|_q^r + \|x_{\tail(2k)}\|_r^r \text{ (by assumption)}\\
{}&\leq 3 \cdot 2^{1-\frac{r}{p}} \left(C'_{p,q}\right)^r k^{1-\frac{r}{s}} \|x_{\tail(k)}\|_s^r + k^{1-\frac{r}{s}}\|x_{\tail(k)}\|_s^r \text{ (Lemma~\ref{lem:shelling})}\\
{}&= \left(1+3 \cdot 2^{1-\frac{r}{p}} \left(C'_{p,q}\right)^r\right) k^{1-\frac{r}{s}} \|x_{\tail(k)}\|_s^r,
\end{align*}
which gives us that
$$\|x - z\|_r \leq C_{p,q,r} \cdot k^{\frac{1}{r}-\frac{1}{s}} \|x_{\tail(k)}\|_s$$
for $C_{p,q,r} = \left(1+3 \cdot 2^{1-\frac{r}{p}} \left(C'_{p,q}\right)^r\right)^{\frac{1}{r}}.$

Now we will prove the case when $r > 1$.  We will use the fact that $\|w - w_B\|_p = \min_{v, \|v\|_0 \leq 2k} \|w - v\|_p$ and in particular that
\begin{equation} \label{eq:w_xa}
\|w - w_B\|_p \leq \|w - x_A\|_p.
\end{equation}
\begin{align*}
\|x - z\|_r &\leq \|x - x_A\|_r + \|x_A - w_B\|_r\text{ (triangle inequality)}\\
&\leq \|x - x_A\|_r + (4k)^{\frac{1}{r}-\frac{1}{p}}\|x_A - w_B\|_p\text{ (Fact~\ref{fact:lp_ab})}\\
&\leq \|x - x_A\|_r + (4k)^{\frac{1}{r}-\frac{1}{p}}\left(\|x - x_A\|_p + \|x - w_B\|_p\right)\text{ (triangle inequality)}\\
&\leq \|x - x_A\|_r + (4k)^{\frac{1}{r}-\frac{1}{p}}\left(\|x - x_A\|_p + \|x - w\|_p + \|w - w_B\|_p\right)\text{ (triangle inequality)}\\
&\leq \|x - x_A\|_r + (4k)^{\frac{1}{r}-\frac{1}{p}}\left(\|x - x_A\|_p + \|x - w\|_p + \|w - x_A\|_p\right)\text{ (\Equation{w_xa})}\\
&\leq \|x - x_A\|_r + (4k)^{\frac{1}{r}-\frac{1}{p}}\left(2\|x - x_A\|_p + 2\|x - w\|_p\right)\text{ (triangle inequality)}\\
&= \|x_{\tail(2k)}\|_r + 2(4k)^{\frac{1}{r}-\frac{1}{p}} \|x_{\tail(2k)}\|_p + 2(4k)^{\frac{1}{r}-\frac{1}{p}} \|x - w\|_p\\
&\leq \|x_{\tail(2k)}\|_r + 2(4k)^{\frac{1}{r}-\frac{1}{p}} \|x_{\tail(2k)}\|_p + 2\cdot4^{\frac{1}{r}-\frac{1}{p}} C'_{p,q} \cdot k^{\frac{1}{r}-\frac{1}{q}} \|x_{\tail(2k)}\|_q\text{ (by assumption)}\\
&\leq \left(1+2\cdot4^{\frac{1}{r}-\frac{1}{p}}\left(1+C'_{p,q}\right) \right)k^{\frac{1}{r}-\frac{1}{s}}\|x_{\tail(k)}\|_s\text{ (Lemma~\ref{lem:shelling})},
\end{align*}
which satisfies our requirement when $C_{p,q,r} = 1+2\cdot4^{\frac{1}{r}-\frac{1}{p}}\left(1+C'_{p,q}\right).$
\end{proof}

Note that this reduction loses a factor of two in $k$ due to its application of Lemma~\ref{lem:shelling}.  In the case of $r \leq 1$, if $q = r = s$ then we don't have to apply Lemma~\ref{lem:shelling} and we can therefore avoid losing this factor of two.  This gives us the following corollary.

\begin{corollary}
For $p, 1 \geq q > 0$, if there exists a recovery scheme $(\Phi, \mathcal{R}')$ such that for some constant $C'_{p,q}$
$$\|x-\mathcal{R}'(\Phi x)\|_p \leq C'_{p,q} \cdot k^{\frac{1}{p}-\frac{1}{q}} \|x_{\tail(k)}\|_q,$$
then there exists a recovery scheme $(\Phi, \mathcal{R})$ such that for some constant $C_{p,q}$
$$\|x-\mathcal{R}(\Phi x)\|_q \leq C_{p,q} \cdot \|x_{\tail(k)}\|_q.$$
\end{corollary}

Note that in this case our recovery algorithm $\mathcal{R}$ projects $\mathcal{R}'(\Phi x)$ to its largest $k$ coordinates, rather than the largest $2k$.

\bibliographystyle{alpha}
\bibliography{biblio}

\end{document}